\documentclass{article}
\usepackage{arxiv}

\usepackage{times}
\usepackage{soul}
\usepackage{url}
\usepackage[hidelinks]{hyperref}
\usepackage[utf8]{inputenc}
\usepackage[small]{caption}
\usepackage{booktabs}
\usepackage{cite}

\usepackage[title]{appendix}
\usepackage{enumerate}
\usepackage{float}
\usepackage{subcaption}
\usepackage{setspace}
\usepackage{xcolor,xparse}
\usepackage{romannum}
\usepackage{tabularx}
\usepackage{array}
\usepackage{mdframed}

\usepackage{amsmath,amsthm,amssymb,amsfonts,bm}
\allowdisplaybreaks
\usepackage{mathtools}
\usepackage{stackrel}                                   
\usepackage{nicefrac}                                   
\usepackage{accents}
\usepackage{scalerel}

\usepackage{algorithmic}
\usepackage[ruled,vlined,linesnumbered]{algorithm2e}    
\newcommand{\nonl}{\renewcommand{\nl}{\let\nl\oldnl}}   
\usepackage{listings}
\SetAlCapNameFnt{\footnotesize}
\SetAlCapFnt{\footnotesize}

\usepackage{dsfont}

\usepackage{graphicx}
\usepackage{longtable}
\usepackage{hhline}
\usepackage{makecell}
\usepackage{varwidth}                                   
\usepackage{wrapfig}                                    
\newcommand{\halfspace}{\kern 0.2em}
\usepackage[explicit,compact]{titlesec}                         

\usepackage{todonotes}

\let\emptyset\varnothing

\usepackage[bottom]{footmisc}
    \DeclareMathOperator*{\argmax}{argmax}                                      
    \DeclareMathOperator*{\argmin}{argmin}                                      
    \DeclareMathOperator*{\maximize}{maximize}  
    \DeclareMathOperator*{\minimize}{minimize}  

    \newcommand{\A}{\mathcal{A}}

    \renewcommand{\H}{\mathcal{H}}

    \renewcommand{\S}{\mathcal{S}}
    
    \newcommand{\X}{\mathcal{X}}

    \newcommand{\vb}{\mathbf{b}}

    \newcommand{\Tf}{T_{\mathrm{f}}}
    
    \usepackage{thmtools}
    \usepackage{thm-restate}

    \newtheorem{theorem}{Theorem} 
    \newtheorem{lemma}[theorem]{Lemma}

    \newtheorem{remark}{Remark}
    \newtheorem{definition}{Definition}

    \graphicspath{{./figures/hier_decept/}}             

\title{Strategic Concealment of Environment Representations in Competitive Games}

\author{
    Yue Guan  \textsuperscript{1} \quad 
    Dipankar Maity \textsuperscript{2}\quad  
    Panagiotis Tsiotras \textsuperscript{1}\\
    \textsuperscript{1}
    Georgia Institute of Technology\\
    \textsuperscript{2}
    University of North Carolina at Charlotte\\
    \texttt{yguan44@gatech.edu, dmaity@charlotte.edu, tsiotras@gatech.edu}
}

\begin{document}

\maketitle

\pagenumbering{arabic}
\begin{abstract}
This paper investigates the strategic concealment of environment representations used by players in competitive games.
We consider a defense scenario in which one player (the Defender) seeks to infer and exploit the representation used by the other player (the Attacker). 
The interaction between the two players is modeled as a Bayesian game:
the Defender infers the Attacker’s representation from its trajectory and places barriers to obstruct the Attacker's path towards its goal, while the Attacker obfuscates its representation type to mislead the Defender.
We solve for the Perfect Bayesian Stackelberg Equilibrium via a bilinear program that integrates Bayesian inference, strategic planning, and belief manipulation.
Simulations show that purposeful concealment naturally emerges:
the Attacker randomizes its trajectory to manipulate the Defender’s belief, inducing suboptimal barrier selections and thereby gaining a strategic advantage.
\end{abstract}


\section{Introduction}
Modern autonomous systems must operate in complex and high-dimensional environments, where maintaining a fully detailed representation of the world is infeasible.
Sensing limitations such as finite resolution and range naturally lead to coarse or partial representations~\cite{lavalle2006planning, thrun2002probabilistic}, 
while computational limits require agents to employ environment representations that deliberately reduce resolution and remove  irrelevant information, for tractable decision-making~\cite{larsson2020q}.
While being indispensable,  such reduced representations inevitably lead to suboptimality in the resulting strategies.

In single-agent settings, such suboptimality is often acceptable as a trade-off for real-time solutions, and prior work has studied the design of such reduced representations~\cite{larsson2020q, baras2000learning}.
In multi-agent environments, however, employing reduced representations is more delicate, since an agent’s reduced representation, if revealed, may expose structural weaknesses that adversaries can exploit~\cite{Waugh2015AbstractionExploitability, guan2022hierarchical}.
This motivates a novel and underexplored problem: 
how should agents behave strategically to conceal their representations, and how can they infer and exploit their opponents' representations?

Motivated by these questions, we investigate strategic representation concealment within a Bayesian game framework. 
We formulate a defense game, where an Attacker reasons using a private environment representation only known to itself and tries to reach a target. 
The Defender with limited prior knowledge of the Attacker's representation needs to select an obstacle barrier to obstruct the Attacker. 
Consequently, the Attacker needs to strategically conceal its representation and intentionally manipulate the Defender's belief to reduce the risk of exploitation.  
To properly counteract, the Defender needs to infer the Attacker's representation from the Attacker's trajectory. 

Our Bayesian formulation offers two key advantages over existing intent inference and concealment frameworks~\cite{ramirez2010probabilistic, pereira2017landmark, price2023domain}. 
First, it captures the notion of \textit{purposeful concealment}—the Attacker conceals its representation only when doing so improves its performance, rather than merely maximizing the uncertainty or the bias in the Defender's belief system. 
Second, the Bayesian approach avoids prescribing a pre-defined Defender inference algorithm, which is often impractical in adversarial settings~\cite{fudenberg1991game}. 

Previous work has developed efficient algorithms for solving Bayesian games using linear programming~\cite{li2014lp}, common-information decomposition~\cite{ouyang2016dynamic}, and sequential rationalization~\cite{rostobaya2025deceptive}.

In those settings, an agent’s type typically corresponds to latent factors such as its goal or unobserved state. 
In contrast, our setting links the Attacker’s type directly to its environment representation, so that the Attacker's type not only shapes rewards but also restricts the set of admissible strategies.
Specifically, we impose that all states within a `superstate' in a given representation share the same action distribution, which is motivated by principles from minimum-attention control~\cite{brockett1997minimum} and hierarchical learning~\cite{baras2000learning}, where decisions are made at the coarse-resolution superstate level rather than at the finest resolution.
As we later show in this work, the interplay between the agent’s type and its feasible strategies significantly increases the problem complexity, leading to an optimization problem subject to bilinear constraints rather than one with linear constraints as in~\cite{li2014lp}.

\paragraph{Contributions.}

We introduce a novel Bayesian game in which a Defender seeks to infer the representation used by an Attacker, and responds by placing a barrier to delay the Attacker’s progress toward a target. 
If the Defender accurately infers the Attacker’s representation, it can exploit it by placing an effective defensive barrier.
This creates a natural incentive for the Attacker to conceal its representation through purposeful behavioral strategies.

We solve for the Perfect Bayesian Equilibrium via a bilinear program that integrates Bayesian inference, strategic planning, and belief manipulation. 
Our results reveal that the Attacker actively introduces ambiguity into the Defender’s Bayesian belief to mislead the Defender. 
Moreover, the bilinear program provides structural insights into the representation concealment problem: 
the Attacker manipulates the Defender’s beliefs by strategically planning a distribution over trajectories in the game tree to manipulate the Defender's belief.

Simulations demonstrate that such concealment yields significant performance gains compared to following an optimal trajectory that directly reveals the environment representation.

\paragraph{Notation.}
We use $\mathcal{P}(E)$ to denote probability distribution over a finite set $E$.
The notation $[T]$ denotes the set of integers $\{0, \ldots, T\}$.
We use $\bm{v} = [v_i]_{i \in \mathcal{I}}$ to denote the vector whose $i$-th component is $v_i$.
The notation $a = (a_{i,j})_{i \in \mathcal{I},\, j \in \mathcal{J}}$ denotes the collection of elements $a_{i,j}$ indexed by $(i,j) \in \mathcal{I} \times \mathcal{J}$.

\section{Problem Formulation}

We consider a two-player defense game played in a finite grid environment, as illustrated in Figure~\ref{fig:representation-concealment-schematic}. 
An Attacker is tasked with reaching a goal location, while the Defender aims to delay the Attacker by placing barriers in the environment. 
At the beginning of the game, the Attacker is privately assigned a \emph{representation} of the environment, with which the Attacker plans a trajectory to its goal.
Specifically, the environment representation partitions the grid into superstates (regions), constraining the Attacker to use the same action distribution for all positions within the same region—effectively reducing its decision-making choices. 
The Defender, who does not observe the Attacker’s assigned representation, seeks to infer it based on the Attacker's behavior and strategically selects a barrier configuration at a designated time $T$ to delay the Attacker's advancement toward the target. 
The strategic interaction between Attacker's representation concealment and the Defender's inference on the representation type is formalized as a two-phase Bayesian game as follows.

\begin{figure}[b]
    \centering
    \includegraphics[width=0.7\linewidth]{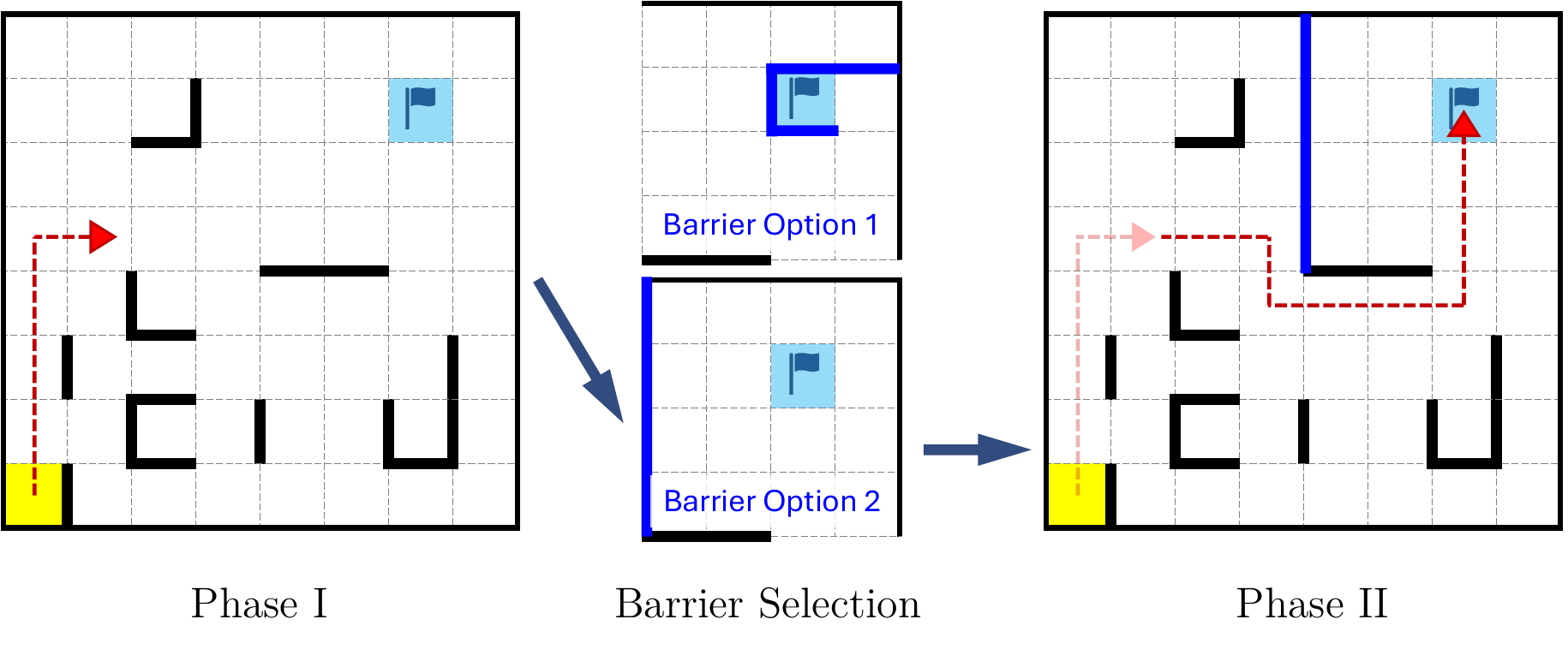}
    \vspace{-0.1in}
    \caption{An example of the proposed two-phase defense game. In Phase~I, the Attacker navigates the grid while the Defender passively observes the state-action trajectory. The middle panel illustrates the two available barrier configurations for the Defender to protect the target. 
    The Defender selects barrier configuration $\omega = 2$ to obstruct the Attacker’s path. In Phase~II, the Attacker continues toward the goal under the modified environment induced by the selected barrier.}
    \vspace{-0.2in}
    \label{fig:representation-concealment-schematic}
\end{figure}

\paragraph{State and action spaces.}
The game state at time $t$ is the Attacker's position $s_t \in \S$, where $\S \subset \mathbb{Z}^2$ is the set of valid grid cells. 
At each time step, the Attacker selects an action $a_t \in \A = \{\texttt{Up,Down,Left,Right}\}$ to transition to $s_{t+1}$, unless blocked by a wall, in which case the Attacker remains at $s_t$.

The Defender selects a barrier configuration $\omega$ at time $T$ from a finite set of allowed configurations $\Omega$. To avoid trivial outcomes, no configuration may fully enclose the goal set $G \subseteq \S$.

\paragraph{Game phases.}
The game unfolds in two sequential phases, separated by a predefined time step $T$.

\begin{itemize}
    \item \textbf{Phase~I} ($t = 0, \ldots, T$): The Attacker moves in the environment, in a manner that is compatible to its assigned representation, while the Defender observes the state-action trajectory.
    \item \textbf{Barrier selection} ($t=T$): After observing the state-action trajectories in Phase~I, the Defender selects and deploys a barrier, thereby altering the traversability of the grid environment. 
    \item \textbf{Phase~II} ($t = T+1, \ldots, \Tf$): 
    With the new barrier added to the environment,  the Attacker continues toward the goal in the new environment until it reaches a goal state $g \in G$. 
    The game terminates at time $\Tf$ when the Attacker reaches a goal state.
\end{itemize}

During Phase~I, the game history is given by $h_t = (s_0, a_0, \ldots, s_{t-1}, a_{t-1}, s_t)$.
For Phase~II, the history is given by $h_t = (s_0, a_0, \ldots, s_{T}, \omega, s_{T+1}, a_{T+1}, \ldots, s_t)$, where $s_{T} = s_{T+1}$ as the Defender is the only player selecting an action at time $T$.
The set of admissible histories at time $t$ is denoted by $\H_t$.
An instance of the proposed two-phase game is presented in Figure~\ref{fig:representation-concealment-schematic}.

\paragraph{Representations.}
A representation $\Gamma$ is a partition of the state space $\S$ into superstates: $\Gamma = \{\gamma^1, \ldots, \gamma^{|\Gamma|}\}$, with $\gamma^i \subseteq \S$ and $\gamma^i \cap \gamma^j = \emptyset$ for $i \neq j$. In this work, we use quad-tree-based representations~\cite{larsson2020q} due to their compatibility with the grid world.
However, our proposed Bayesian framework applies to any form of representation $\Gamma$.

At the beginning of the game, the Attacker is privately assigned a type $\theta \in \Theta$, drawn according to the prior distribution  $\vb_0 \in \mathcal{P}(\Theta)$, where $\theta$ denotes the set of available Attacker types.
Each type $\theta$ uniquely determines a representation $\Gamma^\theta$, 
with the collection of representations given by $(\Gamma^{1}, \ldots, \Gamma^{|\Theta|})$.
A type-$\theta$ Attacker is constrained to select actions such that all states within each superstate $\gamma \in \Gamma^\theta$ share the same action distribution. 
This \textbf{same-action-distribution (SAD)} constraint (formalized later) limits the flexibility of Attacker strategies and introduces opportunities for the Defender to infer and exploit the used representation.

Figure~\ref{fig:example-representations} illustrates how an environment representation can be exploited.
Type-1 representation in subplots (a)–(c) features a coarse resolution near the goal and a fine resolution near the start, whereas type-2 representation in subplots (d)–(e) exhibits the opposite pattern.

Under type-1 representation, barrier option 1 is most effective for the Defender, as it degrades Attacker performance: together with the coarse resolution around the goal, barrier~1 forces the Attacker to rely on a mixed (randomized) strategy to reach the target (dashed trajectory segment).
A poor realization of this mixed strategy is shown in subplot~(b).
In contrast, barrier option 2 is suboptimal for the Defender, as shown in subplot~(c), since it allows deterministic navigation, i.e., a straight path to the goal.

For type-2 representation, the situation reverses.
Barrier option 2 is favorable for the Defender, since it compels the Attacker to adopt a mixed strategy to traverse the fine-resolution region near the goal.
Barrier option 1, however, is suboptimal, as the Attacker can follow straight-line trajectories into the fine-resolution area and then navigate around obstacles with fine maneuvers.

\begin{figure}[b]
    \centering
    \includegraphics[width=\linewidth]{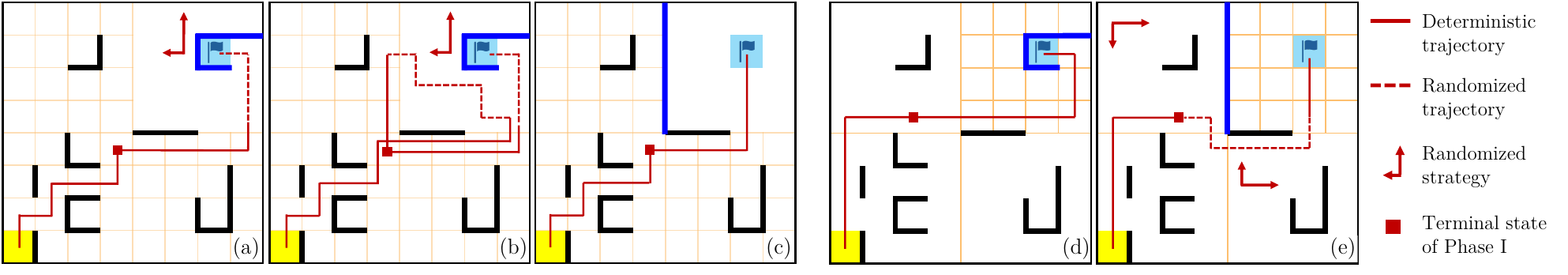}
    \caption{Example of two representations and trajectories under different barrier placement. 
    The solid line shows the trajectory segment that can be achieved with deterministic strategies, while the dashed ones need to be achieved with randomized strategy due to coarse resolution. 
    The randomized strategies are illustrated via the red arrows in the corresponding superstates. 
    Subplots (a) and (b) show different realizations of trajectories from the same mixed strategy.}
    \label{fig:example-representations}
\end{figure}

\paragraph{Information structure.}
The goal set $G$, representations $(\Gamma^\theta)_{\theta=1}^{|\Theta|}$, and barrier configurations $\Omega$ are common knowledge. 
Both players observe the full state-action trajectory, and the Attacker observes the barrier selection $\omega \in \Omega$ at the start of Phase~II.
Crucially, the Attacker's assigned representation is private information for the Attacker, unknown to the Defender.
We further adopt a Stackelberg leader–follower structure: the Attacker (leader) commits to a strategy first, anticipating the Defender’s best response, while the Defender (follower) optimizes given the committed strategy and the induced belief system over the Attacker’s private representation types.

\paragraph{Strategy sets.}
The Attacker’s strategy is a mapping $\pi^\theta_t: \H_t \to \mathcal{P}(\A)$, subject to the \textbf{same-action-distribution} (SAD) constraint: 
for any two histories (in the same phase) ending in states within the same superstate $\gamma \in \Gamma^\theta$, the action distributions must be identical. 
Formally, for Phase~I history $h_t = (s_0, a_0, \ldots, s_t)$ and $h'_\tau = (s'_0, a'_0,\ldots, s'_\tau)$, possibly of different lengths $t, \tau \in [T-1]$, the SAD constraint requires that
\[
    \pi^\theta_t(h_t) = \pi^\theta_\tau(h'_\tau) \qquad \text{if } s_t, s'_\tau \in \gamma, \text{ and } \gamma \in \Gamma^\theta.
\]
Similar SAD constraint applies for the Phase~II histories. 

Since this constraint applies separately in each phase for different history lengths, the Attacker’s strategy can be represented as stationary strategies:
\[
\pi^{\theta}_{\emptyset}, \pi^{\theta}_\omega: \S \to \mathcal{P}(\A),
\]
where $\pi^{\theta}_{\emptyset}$ denotes the type-$\theta$ Attacker’s strategy before the Defender’s barrier selection, and $\pi^{\theta}_{\omega}$ denotes the strategy after barrier $\omega$.
The SAD-constraint requires $\pi^{\theta}_{\emptyset}$ and $\pi^{\theta}_{\omega}$ to satisfy
\begin{equation*}
\pi_\emptyset^\theta(a|s) = \pi_\emptyset^\theta(a|s'), \qquad \pi_\omega^\theta(a|s) = \pi_\omega^\theta(a|s'),
\quad
\forall s, s'\in \gamma,\; \forall a\in\A.
\end{equation*}
Thus, one can equivalently write the stationary strategies as mappings:
$\pi^{\theta}_{\emptyset}, \pi^{\theta}_\omega: \Gamma^\theta \to \mathcal{P}(\A)$.
We denote the type-$\theta$ full strategy across both phases as $\pi^\theta \in \Pi^\theta$, and the complete strategy profile across types as $\pi \!\in\! \Pi$.

The Defender acts only once at $T$, selecting barrier $\omega \in \Omega$ via 
\[
    \sigma:\H_T \to \mathcal{P}(\Omega),
\] 
where $\sigma(\omega|h_T)$ is the probability of choosing barrier $\omega$ given the Phase~I history $h_T$. The Defender’s strategy set is denoted by $\Sigma$.

\begin{remark}
    While the SAD constraint yields stationary Attacker strategies within each phase, the Defender’s strategy remains history-dependent, as the Defender relies on the observed Attacker trajectory to infer the Attacker’s representation and choose a barrier that exploits the SAD-constrained Attacker strategy.
\end{remark}

\paragraph{Optimization and Equilibrium.}
We consider a generic running reward $r(s,a)$, which specifies the instantaneous payoff received by the Attacker when taking action $a$ at state $s$.

Let $(\pi, \sigma)$ denote a pair of Attacker and Defender strategies. 
The expected cumulative reward for a type-$\theta$ Attacker starting at state $s_0$ is given by
\[
    J^\theta(s_0; \pi^\theta, \sigma) = \mathbb{E}_{\pi^\theta, \sigma} \left[ \sum_{t=0}^{\Tf} \beta^t r(s_t, a_t) \big \vert s_0\right],
\]
where $\beta \in (0,1)$ is the discount factor. 
Note that the terminal time $\Tf$ is a random variable that depends on the strategy pair $(\pi^\theta,\sigma)$.

Since the Defender does not observe $\theta$, it evaluates $J^\theta$  with respect to the prior:
\[
    J(s_0; \pi, \sigma, \vb_0) = \sum_{\theta \in \Theta} \vb_0(\theta) J^\theta(s_0; \pi^\theta, \sigma).
\]

Under the Stackelberg game formulation, the Attacker acts as the leader, committing to a strategy first, and the Defender subsequently best responds.
Each type of Attacker maximizes its own objective $J^\theta$, while the Defender minimizes the expected objective $J$.
We define the equilibrium as follows.

\begin{definition}
    An Attacker strategy $\pi^{*} = \big(\pi^{\theta*}\big)_{\theta \in \Theta}$ and a Defender strategy $\sigma^*$ constitutes to a Stackelberg equilibrium if, 
    for all $\pi^\theta \in \Pi^\theta$,
    \begin{subequations}
        \label{eqn:equilibrium-def}
        \begin{align}
            &J^\theta\big(s_0; \pi^{\theta*}, \sigma^*(\pi^*)\big) \geq J^\theta(s_0; \pi^{\theta}, \sigma^*(\pi)), \quad \forall \theta \in \Theta, \\
            &\sigma^*(\pi) \in \argmin_{\sigma \in \Sigma}J(s_0; \pi, \sigma, \vb_0).
        \end{align}
    \end{subequations}
\end{definition}

\section{Phase~II: SAD-Constrained MDP}
\label{sec:cMDP}
Since the game unfolds in two phases, it is natural to adopt a backward induction approach. 
Specifically, we first solve Phase~II, which reduces to a single-agent Markov Decision Process (MDP) given the observed Defender’s barrier placement. 
However, the solution must satisfy the SAD constraint. 
The Phase~II solution, which represents the game’s future outcome, is then used as the boundary conditions in Phase~I. 
We will show that both Phase~I and Phase~II problems can be formulated as bilinear programs, with the bilinear constraints arising from the SAD requirement.

For each Attacker type $\theta \in \Theta$, barrier configuration $\omega \in \Omega$, and Phase~II initial state $s_T \in \S$ (Attacker's position at the end of Phase~I), the Phase~II problem is defined as a constrained MDP given by
\[
\mathrm{cMDP}^\theta_{\omega,s_T} = \langle s_T, \S, \Gamma^\theta, \A, r, f^\omega, \beta \rangle,
\]
where $\S$, $\A$, $r$ and $\beta$ denote the original state space, action space, Attacker reward function, and discount factor, respectively. 
The representation $\Gamma^\theta$ dictates the SAD constraint that the resulting optimal strategy needs to satisfy.
The transition function $f^\omega$ encodes the grid world traversability with the added barriers $\omega$.
Specifically, the transition function is a mapping $f^\omega : \S \times \S \times \A \to [0, 1]$.

\paragraph{Unconstrained MDP.}
Without the SAD constraint, the MDP can be solved via the following bilinear program with the policy $\pi$ and the discounted state occupation measure $d$ as the decision variables, linked by the flow conservation with the initial distribution $\mu$:
\begin{subequations}\label{eq:bilinear-mdp}
\begin{alignat}{2}
\maximize_{\,d,\;\pi}\quad 
& \sum_{s\in\S} d(s)\!\sum_{a\in\A}\! \pi(a| s)\, r(s,a) \\
\mathrm{subject~to}\quad
& d(s) \;=\; \mu(s) + \beta \sum_{s'\in\S}\sum_{a'\in\A} f^{\omega}(s| s',a')\, \pi(a'| s')\, d(s'), 
\qquad && \forall s\in\S, \label{eq:bilinear-flow}\\
& \sum_{a\in\A}\pi(a| s)=1,\quad \pi(a| s)\ge 0, 
\quad && \forall s\in\S,~a\in\A, \\
& d(s)\ge 0, \quad && \forall s\in\S.
\end{alignat}
\end{subequations}
Equivalently, by introducing the state-action occupancy measure
\(\nu(s,a) \coloneqq d(s)\pi(a| s)\),
problem~\eqref{eq:bilinear-mdp} can be reformulated as the following well-known dual linear program for MDPs~\cite{BertsekasTsitsiklis1996}, with $\nu$ as the sole decision variable. 
\begin{subequations} \label{eq100:dual}
\begin{alignat}{2}
    \max_{\nu} \quad & \sum_{s \in \mathcal{S}} \sum_{a \in \mathcal{A}} \nu(s, a) r(s,a) \\
    \mathrm{subject~to} \quad & \sum_{a \in \mathcal{A}} \nu(s, a) = \mu(s) + \beta \sum_{s' \in \mathcal{S}} \sum_{a' \in \mathcal{A}} f^\omega(s | s', a') \nu(s', a'), \quad &&\forall s \in \mathcal{S}, \\
    & \nu(s, a) \geq 0, \quad &&\forall s \in \mathcal{S},~ a \in \mathcal{A}. 
\end{alignat}
\end{subequations}
The optimal solution $(d^*, \pi^*)$ of Problem~\ref{eq:bilinear-mdp} is related to the optimal state-action occupancy measure $\nu^*$ via the following relationships:
\[
d^*(s) = \sum_{a\in \A} \nu^*(s,a),\qquad 
\pi^*(\cdot|s) = 
\begin{cases}
    \frac{\nu^*(s,a)}{d^*(s)} \qquad & \text{ if } d^*(s) >0, \\
    \text{any distribution on $\A$} & \text{ if } d^*(s) = 0.
\end{cases} 
\]

\paragraph{SAD-Constrained MDP.}

Recall that the type-$\theta$ representation $\Gamma^\theta$ induces a disjoint partition on the state space $\S$. 
The SAD constraint requires all states within the same superstate to share a common action distribution.
Formally, for each superstate $\gamma$, the type-$\theta$ Attacker strategy satisfies
\[
\pi^\theta(a|s) =  \pi^\theta(a|s'), \quad \forall s, s'\in \gamma,\; \forall a\in\A.
\]
With $\pi^\theta$ and $d$ as the decision variables, we arrive at the following bilinear program for the SAD-constrained MDPs:
\begin{subequations}\label{eq:sad-bilinear_DM}
\begin{alignat}{2}
\maximize_{\,d,\;\pi^\theta}\quad 
& \sum_{s\in\S} d(s)\!\sum_{a\in\A}\! \pi^\theta(a|s)\, r(s,a) \\
\mathrm{subject~to}\quad
& d(s) \;=\; \mu(s) \;+\; \beta \sum_{s'\in\S}\sum_{a'\in\A} f^{\omega}(s| s',a')\, \pi^\theta(a'|s')\, d(s'), 
\quad && \forall s\in\S, \label{eq:sad-flow}\\
& \sum_{a\in\A}\pi^\theta(a|s)=1,\quad \pi^\theta(a|s)\ge 0, 
\quad && \forall s\in\S,~a\in\A, \\
& d(s)\ge 0, \quad && \forall s\in\S, \\
& \pi^\theta(a|s) =  \pi^\theta(a|s'), \quad &&\forall a\in\A,\; \forall s, s'\in \gamma,\; \forall \gamma\in \Gamma^\theta\; . \label{eq:sad_DM}
\end{alignat}
\end{subequations}

The key difference between Problems~\eqref{eq:bilinear-mdp} and \eqref{eq:sad-bilinear_DM} is the SAD constraint~\eqref{eq:sad_DM}, which encodes how representations restrict feasible strategies.
One can reformulate~\eqref{eq:sad-bilinear_DM} by 
introducing the state–action occupancy measure $\nu^\theta(s,a)$ as we did above in \eqref{eq100:dual},   but unlike the unconstrained case, the problem does not reduce to a linear program because the SAD constraint remains bilinear in terms of $\nu^\theta$.

For unconstrained MDPs, the choice of the initial distribution $\mu$ does not affect the optimal policy, provided that $\mu(s) > 0$ for all $s \in \S$ so that optimality is enforced at every state. 
In contrast, under the SAD constraint the coupling across states makes the optimal solution depend explicitly on $\mu$. 
Intuitively, this arises because the maximization step is no longer performed independently at each state. Instead, a single ‘compromise’ action is chosen based on the average performance across all states in the superstate, weighted by their discounted occupation measure, which in turn depends on the initial distribution.
To obtain optimal values from each Phase~I terminal state $s_T\in\S$, we solve \eqref{eq:sad-bilinear_DM} with $\mu=\mathbf{e}_{s_T}$, which is the one-hot vector for state $s_T$.%
\footnote{
For numerical stability, we add a small constant (e.g., $10^{-5}$) to all states other than $s_T$ in the implementation.}

We solve the bilinear program~\eqref{eq:sad-bilinear_DM} for each representation type 
$\Gamma^\theta$, barrier $\omega \in \Omega$, and Phase~II initial state $s_T \in \S$. 
The resulting optimal value and policy are denoted by $V^{\theta*}_{\omega, s_T}$  and $\pi^{\theta*}_{\omega, s_T}$, respectively. 
Here, $V^{\theta*}_{\omega, s_T}$ is a scalar payoff, while $\pi^{\theta*}_{\omega, s_T}$ is a stationary policy mapping 
$\S \to \mathcal{P}(\A)$. 
The value $V^{\theta*}_{\omega, s_T}$ serves as the terminal payoff of the Phase~I game when the Attacker is of type $\theta$ and the Defender has deployed barrier $\omega$. 
The policy $\pi^{\theta*}_{\omega, s_T}$ specifies the Attacker’s Phase~II strategy that achieves this optimal expected performance if the Attacker lands at $s_T$ at the end of Phase~I.


\section{Phase~I: SAD-Constrained Bayesian Game} 

The key difference between Phase~I and Phase~II lies in the fact that the Defender selects a barrier configuration $\omega \in \Omega$ at time $T$, thereby influencing the outcome of the game. 
Accordingly, Phase~I is modeled as a \emph{game} between the two agents, rather than a single-agent MDP.
Moreover, since the Defender does not observe the representation assigned to the Attacker, Phase~I is an asymmetric-information game~\cite{li2014lp}.

In this setting, the Attacker faces a strategic trade-off between path optimality and the concealment of its representation type. 
For example, the Attacker may deliberately avoid leveraging its fine-resolution representation to navigate around obstacles. 
Instead, it may adopt a suboptimal trajectory that mimics the behavior of a coarser representation in certain regions of the environment. 
The goal is to mislead the Defender into selecting a less effective barrier—one that the Attacker can later easily circumvent.

We formalize this interaction as a \emph{Bayesian game}, wherein the Defender maintains a belief over the Attacker’s representation type and selects the barrier configuration accordingly. 
In turn, the Attacker optimizes its strategy by taking into account the Defender’s belief, balancing path optimality with representation concealment.

The SAD-constrained Bayesian game for Phase~I is given by the following tuple 
\[
    \mathrm{BG} = \big \langle \S, \A, f, r, (V^{\theta*}_{\omega, s})_{s\in\S, \omega \in \Omega}, \Theta, (\Gamma^\theta)_{\theta\in \Theta}, \beta, T \big \rangle,
\]
where $\S$ and $\A$ are the state and action spaces for the Attacker; $f$ is the transition function without the Defender barrier, 
and $r$ is the original reward. 
The terminal values $(V^{\theta*}_{\omega, s})_{s\in\S, \omega \in \Omega}$ come from the Phase~II solution. 
The set $\Theta$ consists of all possible type of the Attackers, and 
$(\Gamma^\theta)_{\theta \in \Theta}$ is the set of representations used by the Attacker. 
The discount factor is denoted as $\beta$ and the terminal time for Phase~I Bayesian game is $T$.

\subsection{Value Functions}

For convenience, we define the following shorthand notations.
Let $\pi^\theta_{h_t} \in \mathcal{P}(\A)$ denote the action distribution of a type-$\theta$ Attacker given the history $h_t$, and let $\pi_{h_t} = (\pi^\theta_{h_t})_{\theta \in \Theta}$ denote the complete strategy profile for all Attacker types with history $h_t$.
Similarly, $\sigma_{h_T} \in \mathcal{P}(\Omega)$ denotes the Defender’s action distribution given the entire Phase~I history $h_T$.
The quantity $b_t^\theta = \vb_t(\theta)$ represents the belief that the Attacker is of type $\theta$ at time $t$.

The type-$\theta$ Attacker value at Phase~I can be easily evaluated by the following recursive formula:
\begin{subequations}
    \begin{align}
        J^\theta_{T}(h_{T},& \sigma_{h_{T}}) = \sum_{\omega \in \Omega}V^{\theta*}_{\omega, s_{T}} \sigma_{h_{T}}(\omega) , \label{eqn:attacker-dp-T}\\
        J^\theta_{T -1}(h_{T - 1}, &\pi_{h_{T - 1}}^\theta) = \sum_{a_{T - 1}\in \A}  \pi_{h_{T - 1}}^{\theta} (a_{T - 1}) \Big(r(s_{T - 1}, a_{T - 1}) + \beta \sum_{s_{T } \in \S}  f(s_{T}|s_{T - 1}, a_{T - 1})  J^\theta_{T}(h_{T}, \sigma_{h_{T}}) \Big), \label{eqn:attacker-dp-T-1} \\
        J^\theta_t(h_t, \pi_{h_t}^\theta) &= \sum_{a_t\in \A}  \pi_{h_t}^{\theta} (a_t) \Big(r(s_t, a_t) + \beta \sum_{s_{t+1} \in \S}  f(s_{t+1}|s_t, a_t)  J^\theta_{t+1}(h_{t+1},\pi_{h_{t+1}}^\theta) \Big), \quad \qquad \forall t \in [T-2], \label{eqn:attacker-dp-t}
    \end{align}
\end{subequations}

where $h_{t+1} = (h_t, a_t, s_{t+1})$ and $V^{\theta*}_{\omega, s}$ is optimal value from the solution of the Phase~II SAD-constrained MDP in Section~\ref{sec:cMDP}.

Next, we derive the recursive formula for the Defender value.
For the terminal time step, we have, by definition,
\begin{equation}
    J_{T}(h_{T}, \vb_{T}, \sigma_{h_{T}}) = \sum_{\theta  \in \Theta} b_T^\theta J^\theta_{T}(h_{T}, \sigma_{h_{T}}) = \sum_{\theta \in \Theta} b^\theta_{T} \sum_{\omega \in \Omega}V^{\theta*}_{\omega, s_{T}} \sigma_{h_{T}}(\omega). \label{eqn:defender-dp-T}
\end{equation}

For time step $t \in [T-2]$, define the reward vector $\bm{r}(s) = [r(s, a)]_{a\in \A}$, then we have that
\begin{align}
    &J_t(h_t, \vb_t, \pi_{h_t}) = \sum_{\theta  \in \Theta} b_t^\theta J^\theta_t(h_t, \pi_{h_t}^\theta) \nonumber\
    \\
    &= \sum_{\theta \in \Theta} b^\theta_t \Big(\bm{r}(s_t)^\top \pi^\theta_{h_t}  + \beta \sum_{s_{t+1} \in \S} \sum_{a_t\in \A}  \pi_{h_t}^{\theta} (a_t) f(s_{t+1}|s_t, a_t)  J^\theta_{t+1}(h_{t+1},\pi_{h_{t+1}}^\theta)\Big) \nonumber
    \\
    &= \sum_{\theta \in \Theta} b^\theta_t \bm{r}(s_t)^\top \pi^\theta_{h_t} +\beta \sum_{s_{t+1} \in \S} \sum_{a_t\in \A} f(s_{t+1}|s_t, a_t) \Big(\sum_{\theta \in \Theta} b_t^\theta \pi_{h_t}^{\theta} (a_t) J^\theta_{t+1}(h_{t+1},\pi_{h_{t+1}}^\theta) \Big) \nonumber
    \\
    &= \sum_{\theta \in \Theta} b^\theta_t \bm{r}(s_t)^\top \pi^\theta_{h_t}  + \beta \sum_{s_{t+1} \in \S} \sum_{a_t\in \A} f(s_{t+1}|s_t, a_t) \underbrace{\bar \chi(\vb_t, \pi_{h_t}, a_t)}_{=\sum_{\theta}b_t^\theta \pi_{h_t}^{\theta} (a_t)} \Big(\sum_{\theta \in \Theta} \underbrace{\frac{b_t^\theta \pi_{h_t}^{\theta} (a_t)}{\bar \chi(\vb_t, \pi_{h_t}, a_t)}}_{b^{\theta+}(\vb_t, \pi_{h_t}, a_t)} J^\theta_{t+1}(h_{t+1},\pi_{h_{t+1}}^\theta) \Big) \nonumber
    \\
    &= \sum_{\theta \in \Theta} b^\theta_t \bm{r}(s_t)^\top \pi^\theta_{h_t} + \beta \sum_{s_{t+1} \in \S} \sum_{a_t\in \A} f(s_{t+1}|s_t, a_t) \bar \chi(\vb_t, \pi_{h_t}, a_t) \Big(\underbrace{\sum_{\theta \in \Theta} b^{\theta+}(\vb_t, \pi_{h_t}, a_t) J^\theta_{t+1}(h_{t+1},\pi_{h_{t+1}}^\theta)}_{\triangleq J_{t+1}(h_{t+1}, \vb^{+}(\vb_t, \pi_{h_t}, a_t), \pi_{h_{t+1}})} \Big)\nonumber
    \\
    &=\sum_{\theta \in \Theta} b^\theta_t \bm{r}(s_t)^\top \pi^\theta_{h_t} + \beta \sum_{s_{t+1} \in \S} \sum_{a_t\in \A} f(s_{t+1}|s_t, a_t) \bar \chi(\vb_t, \pi_{h_t}, a_t) J_{t+1}(h_{t+1}, \vb^+(\vb_t, \pi_{h_t}, a_t), \pi_{h_{t+1}}), \label{eqn:defender-dp-t}
\end{align}
where we use notations $\vb_t = [b^\theta_t]_{\theta \in \Theta}$ and $\vb^+(\cdot) = \big(b^{\theta+}(\cdot)\big)_{\theta \in \Theta}$ for the belief update performed on each type.

Similarly, for time $T-1$ we have
\begin{align*}
     J_t(h_{T-1}, \vb_{T-1}, &\pi_{h_{T-1}})
    = \sum_{\theta \in \Theta} b^\theta_{T-1} \bm{r}(s_{T-1})^\top \pi^\theta_{h_{T-1}} \\
    &+ \beta \hspace{-0.05in}\sum_{s_{T} \in \S} \sum_{a_{T-1}\in \A} \hspace{-0.05in}f(s_{T}|s_{T-1}, a_{T-1}) \bar \chi(\vb_{T-1}, \pi_{h_{T-1}}, a_{T-1}) J_{T}(h_{T}, \vb^+(\vb_{T-1}, \pi_{h_{T-1}}, a_{T-1}), \sigma_{h_T})
\end{align*}

The term $\vb^+(\vb_t, \pi_{h_t}, a_t) = \Big[ b^{\theta+}(\vb_t, \pi_{h_t}, a_t) \Big]_{\theta \in \Theta} \in \mathcal{P}(\Theta)$, is the Bayesian belief update rule for the Defender. Specifically, we have
\begin{equation}
    \label{eqn:belief-update}
    b^\theta_{t+1} = b^{\theta+}(\vb_t, \pi_{h_t}, a_t) = \frac{b_t^\theta \pi_{h_t}^{\theta} (a_t)}{\bar \chi(\vb_t, \pi_{h_t}, a_t)} = \frac{b_t^\theta \pi_{h_t}^{\theta} (a_t)}{\sum_{\theta'} b_t^{\theta'} \pi_{h_t}^{\theta'} (a_t)},
\end{equation}
where $\bar \chi(\vb_t, \pi_{h_t}, a_t)$ corresponds to the total probability that the Attacker selects action $a_t$ given Defender's current  belief $\vb_t$, the Attacker strategy $\pi_{h_t}$ and the current history $h_t$.
The numerator $b_t^\theta \pi_{h_t}^{\theta} (a_t)$ gives the probability that action $a_t$ is selected by a type-$\theta$ Attacker.
Consequently, $\vb^{+}(\vb_t, \pi_{h_t}, a_t)$ representing the updated belief about the Attacker's type given the observed action $a_t$.
For a detailed discussion of the derivation of the update rule in~\eqref{eqn:belief-update}, refer to Appendix~\ref{appdx-sec:belief-update}.

\begin{remark}
    For clarity of presentation, the derivation in~\eqref{eqn:defender-dp-t} omits the case where $\bar \X(\vb_t, \pi_{h_t}, a_t) = 0$. 
    As discussed later in Remark~\ref{rmk:off-equilibrium}, this case corresponds to off-equilibrium paths and thus does not directly affect the equilibrium outcome. 
    The division by $\bar \chi$ is used only to normalize the belief in constructing the dynamic program. 
    When constructing the solution, however, we multiply $\bar \chi$ back with the $\vb^+$ term, so the normalization step is not necessary. 
\end{remark}

\subsection{Perfect Bayesian Stackelberg Equilibrium (PBSE)}

With asymmetric information, the Defender must update its belief about the Attacker’s type from observed trajectories, while the Attacker, though knowing its own type, must account for the Defender’s evolving belief when planning its actions.
This interplay is captured by the Perfect Bayesian Stackelberg equilibrium (PBSE), which requires agents to act rationally given the game state and their beliefs, while ensuring consistency between strategies and Bayesian belief updates.

\begin{definition}
    \label{def:strong-PBSE}
    A strategy profile $(\pi^*, \sigma^*)$ forms a Perfect Bayesian Stackelberg Equilibrium (PBSE) if the following holds for all histories $h_t \in \H_t$, and for all time steps $t \in [T]$. 
\begin{enumerate}
    \item[a)] $J^\theta_t(h_t, \pi_{h_t}^{\theta*}) \geq J^\theta_t(h_t, \pi_{h_t}^{\theta})$ for all $\theta \in \Theta$ and $\pi_{h_t}^\theta \in \mathcal{P}(\A)$ and $t\in [T -1]$;
    \item[b)] $\sigma^*_{h_{T}} (\pi) \in \argmin_{\sigma \in \Sigma} J_{T}(h_{T}, \vb_{T}(\pi), \sigma_{h_{T}})$ ;
    \item[c)] The belief $\vb_{T}$ above is propagated with Bayes' rule under the history $h_{T}$ conditioned on the strategy $\pi^{*}$.
\end{enumerate}
\end{definition}

Note that the PBSE is defined via the optimality for each type of Attacker, and thus the optimization objective for the Attacker differs from that of the Defender.
To obtain a unified max-min dynamic program, we present the following proposition.
\begin{restatable}{proposition}{eqv}
    \label{prop:equivalence}
    If $\vb_t > \mathbf{0}$ element-wise, then 
    \begin{alignat}{2}
        &J^\theta_t(h_t, \pi_{h_t}^{\theta*}) \geq J^\theta_t(h_t, \pi_{h_t}^{\theta}), \qquad && \forall \theta \in \Theta, \pi^\theta_{h_t} \in \mathcal{P}(\A), \nonumber \\
        \makebox[0pt][l]{\hspace*{-1.1in} if and only if }
        &\sum_{\theta} b_t^\theta J^\theta_t(h_t, \pi_{h_t}^{\theta*}) \geq \sum_{\theta} b_t^\theta  J^\theta_t(h_t, \pi_{h_t}^{\theta}), \qquad && \forall \pi_{h_t} = \big(\pi_{h_t}^\theta\big)_{\theta \in \Theta} \in \big(\mathcal{P}(\A)\big)^{|\Theta|}. \label{eqn:averaged-condition}
    \end{alignat}
\end{restatable}

\begin{proof}
    See Appendix~\ref{appdx-sec:prop-1}.
\end{proof}

\begin{remark}
    \label{rmk:off-equilibrium}
    There are two cases in which $b^\theta_t = 0$:
    \begin{enumerate}[i)]
        \item Prior $b^\theta_0 = 0$ for some type $\theta$. 
        Then, the belief on $\theta$ remains zero throughout the game. 
        In this degenerate case, $\theta$ can be excluded from the set of candidate representations without loss of generality.
        Hence, we assume $b^\theta_0 > 0$ for all $\theta \in \Theta$.
        \item Prior $b^\theta_0 > 0$ for all $\theta \in \Theta$, but the realized trajectory $h_t$ results in $b^\theta_t = 0$. 
        From the belief update rule~\eqref{eqn:belief-update}, one concludes that
        such a trajectory has zero probability under the policy $\pi_{h_t}^\theta$. 
        If $\pi_{h_t}^\theta$ is the considered equilibrium strategy, it follows that the trajectory $h_t$ is an off-equilibrium path; that is, under equilibrium play, such a path is realized with zero probability.
    \end{enumerate}
\end{remark}

In light of Proposition~\ref{prop:equivalence}, the per-type condition (\textit{a}) on $J^\theta_t$ in Definition~\ref{def:strong-PBSE} can be replaced by the averaged condition in~\eqref{eqn:averaged-condition}. 
However, this substitution only considers belief update along trajectories with nonzero probability under the equilibrium strategies.

\begin{definition}
    \label{def:weak-PBSE}
    A strategy pair $(\pi^*, \sigma^*)$ forms a \emph{weak} Perfect Bayesian Stackelberg Equilibrium if the following holds for all histories $h_t \in \H_t$ with non-zero probability under $(\pi^*, \sigma^*)$.
\begin{enumerate}
    \item[a)] $J_t(h_t, \vb_t, \pi^*_{h_t}) \geq J_t(h_t, \vb_t, \pi_{h_t})$ for $\pi_{h_t} \in \Big(\mathcal{P}(\A)\Big)^{|\Theta|}$ and $t\in [T -1]$;
    \item[b)] $\sigma^*_{h_{T}} (\pi) \in \argmin_{\sigma \in \Sigma} J_{T}(h_{T}, \vb_{T}(\pi), \sigma_{h_{T}})$ ;
    \item[c)] The beliefs $\vb_{t}$ are propagated with Bayes' rule under the history $h_{t}$ conditioned on the strategy $\pi^{*}$.
\end{enumerate}
\end{definition}

\subsection{Optimal Value Functions}
Definition~\ref{def:weak-PBSE} unifies the optimization objectives of the Attacker and Defender, yielding a min–max zero-sum formulation amenable to a dynamic programming solution.

\begin{align}
    &\underline{J}_T^*(h_T, \vb_T) = \min_{\sigma_{h_T}} \sum_{\theta \in \Theta} b^\theta_T {\bm{V}^{\theta}_{s_{T}}}^\top \sigma_{h_T},\label{eqn:max-min-dp-T} \\
    &\underline{J}^*_t(h_t, \vb_t) = \max_{\pi_{h_t}} ~~ \sum_{\theta \in \Theta} b^\theta_t \bm{r}(s_t)^\top \pi^\theta_{h_t}
    + \beta \hspace{-0.1in}\sum_{s_{t+1} \in \S} \sum_{a_t\in \A} f(s_{t+1}|s_t, a_t) \bar \chi(\vb_t, \pi_{h_t}, a_t) \underline{J}^*_{t+1}(h_{t+1}, \vb^+(\vb_t, \pi_{h_t}, a_t)),\label{eqn:max-min-dp-t}
\end{align}
where $\bm{V}^{\theta*}_{s_{T}} = [V^{\theta*}_{\omega, s_{T}}]_{\omega \in \Omega} \in \mathbb{R}^{|\Omega|}$ is the vector of the terminal values of Phase~I. 

Finally, we present the following scaling property of the optimal value function that is key to the derivation of the bilinear program solution of the above min-max optimization problem. 

\begin{lemma}
    \label{lmm:linear-value}
    For any constant $\alpha >0$ and time step $t \in [T]$, the game value $J_t^*(h_t, \vb_t)$ satisfies 
    \[J_t^*(h_t, \alpha \vb_t) = \alpha J_t^*(h_t, \vb_t).\]
\end{lemma}
\begin{proof}
    The result follows directly from the definition $J_t(h_t, \vb_t, \pi_{h_t}) = \sum_{\theta  \in \Theta} \vb_t(\theta) J^\theta_t(h_t, \pi_{h_t}^\theta)$ for $t \in [T-1]$, and $J_T(h_T, \vb_T, \sigma_{h_t}) = \sum_{\theta  \in \Theta} b^\theta_T  {\bm{V}^{\theta}_{s_{T}}}^\top \sigma_{h_T}$. 
\end{proof}

\subsection{Linear-Program Solution without SAD Constraint}
We start with the optimization problem that provides the optimal strategy of the Attacker.
We will first ignore the SAD constraint due to representations and focus on the belief propagation and the resulting $T$-step optimization, after which we will add the SAD constraints.

\paragraph{Terminal time \texorpdfstring{$T$}{T}.}
At the Phase~I terminal time $T$, the Defender performs the following minimization problem.
\begin{alignat*}{3}
    J_{T}^*(h_{T}, \vb_{T}) &= \quad \minimize_{\sigma_{h_{T}}}  \quad && \sum_{\theta \in \Theta} b_T^\theta {\bm{V}^{\theta}_{s_{T}}}^\top \sigma_{h_T} \\
    &  \qquad \mathrm{subject}~\mathrm{to} ~~ && \mathbf{1}^\top \sigma_{h_T} = 1, ~~ \sigma_{h_T} \geq \mathbf{0}.
\end{alignat*}
\noindent
Through the standard duality result~\cite{bertsimas1997introduction}, we have the following equivalent linear optimization problem. 
\begin{equation}
    \label{eqn:primal-T}
    \begin{alignedat}{3}
    J_{T}^*(h_{T}, \vb_{T}) &= \quad \maximize_{z_{h_T}, \ell_{h_{T}}}  \quad && \ell_{h_T} \\
    &  \qquad \mathrm{subject}~\mathrm{to} ~~ && \sum_{\theta \in \Theta} z_{h_T}^\theta \bm{V}^{\theta}_{s_{T}} \geq \ell_{h_T} \bm{1}, \\
    &       && z_{h_T}^\theta = b_T^\theta,  \qquad\qquad\qquad \qquad  && \forall \theta \in \Theta,
\end{alignedat}
\end{equation}
where for consistency with later derivations, we make the change of variable $z_{h_T}^\theta = b_T^\theta$, and use the shorthand notation $z_{h_T} = (z^\theta_{h_T})_{\theta \in \Theta}$.

\paragraph{Time step \texorpdfstring{$T-1$}{T-1}.}
Now consider the maximization problem performed by the Attacker at time $T-1$. From Lemma~\ref{lmm:linear-value}, it follows that
\begin{align*}
    &J^{*}_{T-1}(h_{T-1}, \vb_{T-1})   \\
    &= \max_{\pi_{h_{T-1}}} \sum_{\theta \in \Theta} b^\theta_{T-1} \bm{r}(s_{T-1})^\top \pi^\theta_{h_{T-1}} \\
    & \qquad \qquad\qquad\qquad + \beta  \sum_{s_{T}} \sum_{a_{T-1}} f(s_{T}|s_{T-1}, a_{T-1}) \bar \chi(\vb_{T-1}, \pi_{h_{T-1}}, a_{T-1}) {J}^*_{T}(h_{T}, \vb^+(\vb_{T-1}, \pi_{h_{T-1}}, a_{T-1})) \\
    &= \max_{\pi_{h_{T-1}}} \sum_{\theta \in \Theta} b^\theta_{T-1} \bm{r}(s_{T-1})^\top \pi^\theta_{h_{T-1}} +  \sum_{s_{T}} \sum_{a_{T-1}}  {J}^*_{T}(h_{T}, \beta f_{s_{T},s_{T-1}}^{a_{T-1}} \bar \chi(\vb_{T-1}, \pi_{h_{T-1}}, a_{T-1}) \vb^+(\vb_{T-1}, \pi_{h_{T-1}}, a_{T-1})) \\
    &= \max_{\pi_{h_{T-1}}} \sum_{\theta \in \Theta} b^\theta_{T-1} \bm{r}(s_{T-1})^\top \pi^\theta_{h_{T-1}} +  \sum_{s_{T}} \sum_{a_{T-1}}  \underbrace{{J}^*_{T}\Big(h_{T},  \big[\beta b_{T-1}^\theta f^{a_{T-1}}_{s_T, s_{T-1}} \pi^\theta_{h_{T-1}}(a_{T-1})\big]_{\theta \in \Theta}\Big)}_{A\big(\pi_{h_{T-1}}(a_{T-1})\big)},
\end{align*}
where $h_T = (h_{T-1}, a_{T-1}, s_T)$ is the new history induced after applying action $a_{T-1}$ and transitioning to $s_T$.
We use $f_{s_{T},s_{T-1}}^{a_{T-1}}$ to denote the transition probability $f(s_{T}|s_{T-1}, a_{T-1})$.
The policy $\pi_{h_{T-1}}$ also needs to satisfy the standard probability constraints as well as the SAD constraint.

We first analyze the term $A$, the optimal terminal value under the belief induced by policy $\pi_{h_{T-1}}$ and selected action $a_{T-1}$. 
We evaluate the optimal value under a fixed choice of $\pi_{h_{T-1}}$. 
Plugging in the induced (un-normalized) belief into~\eqref{eqn:primal-T}, we have

\begin{equation*}
    \begin{alignedat}{3}
    A\big(\pi_{h_{T-1}} (a_{T-1})\big) &= \quad \maximize_{\substack{z_{h_T}, \pi_{h_{T-1}},\\ \ell_{h_T}}}  \quad && \ell_{h_T} \\
    &  \qquad \mathrm{subject}~\mathrm{to} ~~ && \sum_{\theta \in \Theta} z_{h_T}^\theta \bm{V}^{\theta}_{s_{T}} \geq \ell_{h_T} \bm{1}, \\
    &           && z_{h_T}^\theta = \beta f^{a_{T-1}}_{s_T, s_{T-1}}  b_{T-1}^\theta \pi^\theta_{h_{T-1}}(a_{T-1}), \qquad \qquad && \forall \theta \in \Theta, \\
    &           && \bm{1}^\top \pi^\theta_{h_{T-1}} = 1, ~~ \pi^\theta_{h_{T-1}} \geq 0, && \forall \theta \in \Theta.
\end{alignedat}
\end{equation*}
For clarity of presentation, we again perform the change of variables $z^\theta_{h_{T-1}} = b_{T-1}^\theta \pi^\theta_{h_{T-1}}$,
which represents the {history-action occupation measures}.
Under this change of variables,  the optimization problem becomes:
\begin{equation}
    \label{eqn:primal-A-T-1}
    \begin{alignedat}{3}
    A\big(z_{h_{T-1}} (a_{T-1})\big) &= \quad \maximize_{\substack{z_{h_T}, z_{h_{T-1}},\\ \ell_{h_T}}}  \quad && \ell_{h_T} \\
    &  \qquad \mathrm{subject}~\mathrm{to} ~~ && \sum_{\theta \in \Theta} z_{h_T}^\theta \bm{V}^{\theta}_{s_{T}} \geq \ell_{h_T} \bm{1}, \\
    &           && z_{h_T}^\theta =  \beta f^{a_{T-1}}_{s_T, s_{T-1}} z^\theta_{h_{T-1}}(a_{T-1}),\qquad \qquad && \forall \theta \in \Theta,\\
    &           && \mathbf{1}^\top z^\theta_{h_{T-1}} = b_{T-1}^\theta, ~~ z^\theta_{h_{T-1}} \geq 0,  \qquad && \forall \theta \in \Theta.
\end{alignedat}
\end{equation}

Leveraging the new formulation in~\eqref{eqn:primal-A-T-1}, and maximizing with respect to the variable $z_{h_{T-1}} = (z_{h_{T-1}}^\theta)_{\theta}$, we obtain the following maximization problem for $J^*_{T-1}$.
Noting that the summations over $a_{T-1}$ and $s_T$ are equivalent to summing over all feasible successor histories $h_T$ starting from $h_{T-1}$, we obtain

\begin{subequations}
    \begin{alignat}{3}
    J^{*}_{T-1}(h_{T-1}, \vb_{T-1}) &= \quad \maximize_{\substack{z_{h_T}, z_{h_{T-1}},\\ \ell_{h_T}, \ell_{h_{T-1}}}}  \quad && \ell_{h_{T-1}} ~+ \sum_{h_{T} \in \H_{T}}\ell_{h_T} \nonumber \\
    &  \qquad \mathrm{subject}~\mathrm{to} ~~ && \ell_{h_T} \bm{1} \leq \sum_{\theta \in \Theta} \bm{V}^{\theta}_{s_{T}} z_{h_T}^\theta , && \forall h_T \in \H_T, \label{eqn:primal-T-1-a}
    \\
    &           && z_{h_T}^\theta =  \beta f^{a_{T-1}}_{s_T, s_{T-1}} z^\theta_{h_{T-1}}(a_{T-1}), ~~\qquad \qquad \qquad && \forall \theta \in \Theta,~ \forall h_T \in \H_T,
    \label{eqn:primal-T-1-b}
    \\
    & && \ell_{h_{T-1}} = \sum_{\theta\in \Theta}  \bm{r}(s_{T-1})^\top z_{h_{T-1}}^\theta , \label{eqn:primal-T-1-c}
    \\
    &           && \mathbf{1}^\top z^\theta_{h_{T-1}} = b_{T-1}^\theta, ~~~ z^\theta_{h_{T-1}} \geq 0,   \qquad && \forall \theta \in \Theta.     \label{eqn:primal-T-1-d}
\end{alignat}
\end{subequations}

Constraint~\eqref{eqn:primal-T-1-a} corresponds to the optimization performed by the Defender at time $T$; 
constraint~\eqref{eqn:primal-T-1-b}  captures the influence of the policy $z_{h_{T-1}}$ on the belief at time $T$, and thus affects the optimization at the next stage of the game;
in constraint~\eqref{eqn:primal-T-1-c},
we have $\bm{r}(s_{T-1}) = [r(s_{T-1}, a)]_{a \in \A}$;
finally, constraint~\eqref{eqn:primal-T-1-d}  encodes the starting belief at $T-1$.

\paragraph{Generic primal problem.}
For the entire $T$-stage game, the primal optimization problem is given by 
\begin{equation}
    \label{eqn:primal}
    \begin{alignedat}{3}
        J^{*}(s_0, \vb_{0})  &= \maximize_{\substack{z_{h_t}, \ell_{h_t}\\
        \forall t\in [T]}}  \quad && \sum_{t=0}^T \sum_{h_{t} \in \H_{t}} \ell_{h_t}
        \\
        & \qquad \mathrm{subject}~\mathrm{to}~~~ && \ell_{h_0} =\sum_{\theta\in \Theta}  \bm{r}(s_{0})^\top z_{h_{0}}^\theta ,
        \\
        & &&\mathbf{1}^\top z^\theta_{h_0} = b^\theta_0, ~~ z^\theta_{h_0} \geq \mathbf{0}, ~~ &&\forall~ \theta\in \Theta,
        \\[0.15in]
        & &&\forall \; t = 1, \ldots, T-1, ~ h_t \in \H_t \\
        & &&\ell_{h_t} =\sum_{\theta\in \Theta}  \bm{r}(s_{})^\top z_{h_{t}}^\theta , \\
        & &&\mathbf{1}^\top z^\theta_{h_t} = \beta f_{s_{t},s_{t-1}}^{a_{t-1}} z^\theta_{h_{t-1}}(a_{t-1}), \qquad  \qquad && \forall~ \theta\in \Theta, 
        \\[0.15in]
        &              &&\ell_{h_T} \bm{1} \leq \sum_{\theta \in \Theta} z_{h_T}^\theta \bm{V}^{\theta}_{s_{T}}  , && \forall h_T \in \H_T,\\
    &           && z_{h_T}^\theta = \beta  f^{a_{T-1}}_{s_T, s_{T-1}} z^\theta_{h_{T-1}}(a_{T-1}), \qquad \qquad \qquad && \forall \theta \in \Theta,~ \forall h_T \in \H_T.\\
    \end{alignedat}
\end{equation}
In the above optimization problem, the $\ell$ variables encodes the discounted \emph{stage reward} for each trajectory, and the $z$ variables are the discounted \emph{history-action occupation measures},
analogous to the state-action occupation measures in MDPs.
Let $z^{\theta*}_{h_t}$ denote the optimal solution to the linear program in~\eqref{eqn:primal}, the optimal Attacker strategy is given by 
\begin{equation}
    \label{eqn:optimal-attacker-strategy}
    \pi_{h_t}^{\theta*}(u) = \left \{
    \begin{array}{ll}
        \frac{z^{\theta*}_{h_t}(u)}{\sum_{\bar u}z^{\theta*}_{h_t}(\bar u)}, \quad &  \text{if } \sum_{\bar u}z^{\theta*}_{h_t}(\bar u) > 0,\vspace{+0.1in} \\
        \frac{1}{|\A(h_t)|}, & \text{otherwise.}
    \end{array}
    \right .
\end{equation}
Note that if $\sum_{\bar u} z^{\theta*}_{h_t}(\bar u) = 0$, the trajectory $h_t$ lies off the equilibrium path. For practical purpose, the strategies on such off-equilibrium trajectories are inconsequential and we simply set them to be uniform.

\subsection{Bilinear-Program Solution}
By augmenting the SAD constraint to the linear program derived above, we arrive at the following bilinear program. 
\begin{equation}
    \label{eqn:bilinear}
    \begin{alignedat}{3}
        J^{*}(s_0, \vb_{0})  &= \maximize_{\substack{z_{h_t}, \ell_{h_t}\\
        \forall t\in [T]}}  \quad \quad && \sum_{t=0}^T \sum_{h_{t} \in \H_{t}} \ell_{h_t}
        \\
        & \qquad \mathrm{subject}~\mathrm{to}~~~ && \ell_{h_0} =\sum_{\theta\in \Theta}  \bm{r}(s_{0})^\top z_{h_{0}}^\theta ,
        \\
        & &&\mathbf{1}^\top z^\theta_{h_0} = b^\theta_0, ~~ z^\theta_{h_0} \geq \mathbf{0}, ~~ &&\forall~ \theta\in \Theta,
        \\[0.15in]
        & &&\forall \; t = 1, \ldots, T-1, ~ h_t \in \H_t \\
        & &&\ell_{h_t} =\sum_{\theta\in \Theta}  \bm{r}(s_{})^\top z_{h_{t}}^\theta , \\
        & &&\mathbf{1}^\top z^\theta_{h_t} = \beta f_{s_{t},s_{t-1}}^{a_{t-1}} z^\theta_{h_{t-1}}(a_{t-1}), \qquad \qquad && \forall~ \theta\in \Theta, 
        \\[0.15in]
        &              &&\ell_{h_T} \bm{1} \leq \sum_{\theta \in \Theta} z_{h_T}^\theta \bm{V}^{\theta}_{s_{T}}  , && \forall h_T \in \H_T,\\
    &           && z_{h_T}^\theta = \beta  f^{a_{T-1}}_{s_T, s_{T-1}} z^\theta_{h_{T-1}}(a_{T-1}), \qquad \qquad \qquad && \forall \theta \in \Theta,~ \forall h_T \in \H_T,\\
    & &&\forall \; t, \tau \in [0, T-1] \\
    &           && \bar z^\theta_{h'_\tau} z^\theta_{h_t}  = \bar z^\theta_{h_t} z^\theta_{h'_\tau}, && \text{if } s_t, s'_\tau \in \gamma, \\
    &   && \bar z^\theta_{h_t} = \sum_{a \in \A} z^\theta_{h_t}(a) && \forall h_t \in \H_t.
    \end{alignedat}
\end{equation}
The optimal SAD-constrained Attacker strategy is similarly given by~\eqref{eqn:optimal-attacker-strategy}.

\paragraph{Defender's strategy.}
Given a complete Phase~I history $h_T$ and the optimal Attacker strategy $\pi^*$, the Defender can compute its belief on the Attacker's representation type using the belief update rule in~\eqref{eqn:belief-update}.
Suppose the belief at time $T$ is $\vb_T$, the Defender's strategy is then given by 
\begin{equation}
    \label{eqn:defender-strategy}
    \sigma^*_{h_T} \in \argmax_{\sigma \in \mathcal{P}(\Omega)} \sum_{\theta \in \Theta} b^\theta_T {\bm{V}_{s_T}^\theta}^\top \sigma.
\end{equation}

\begin{remark}
    The optimization problem in~\eqref{eqn:bilinear} shows that the Attacker optimizes over a distribution of trajectories across the game tree. Since the Attacker has full information and controls how the Defender's belief evolves through Attacker's actions, the optimization effectively plans over the induced belief dynamics by shaping the trajectory distribution.
\end{remark}

\section{Numerical Examples}

We present the solution to the example shown in Figure~\ref{fig:representation-concealment-schematic}, where the Phase~I terminal time is $T = 6$, and the initial belief is set to uniform: $\vb_0 = [0.5, 0.5]$. 
The running reward is defined as 
\begin{equation}
    \label{eqn:concealment-reward}
    r(s,a) = 
    \begin{cases}
        -0.1, & \text{if } s \notin G,  \\
        1.0,  & \text{if } s \in G,
    \end{cases}
\end{equation}
where the Attacker incurs a step cost of $-0.1$ until reaching a goal state $s \in G$, at which point it receives a terminal reward of $1.0$.
Figure~\ref{fig:example-trajectory} illustrates the Phase~I trajectories of the Attacker for each representation type.

\begin{figure}[t]
    \centering
    \includegraphics[width=0.9\linewidth]{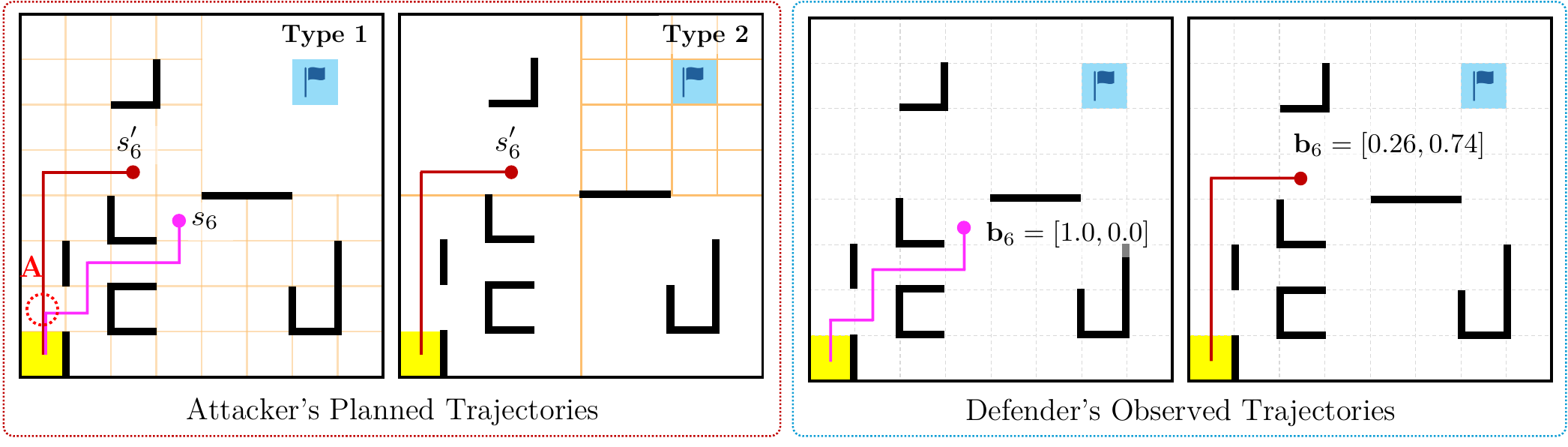}
    \caption{Left: Attacker trajectories under each representation type. 
    Right: Observed trajectories from the Defender’s perspective.}
    \label{fig:example-trajectory}
    \vspace{-0.1in}
\end{figure}

A key observation is that the type-1 Attacker randomizes its action at Point A: choosing \texttt{Up} with probability 0.358 and \texttt{Right} with 0.642. 
By selecting \texttt{Up}, the type-1 Attacker mimics the behavior of a type-2 Attacker which has a coarse representation in the obstacle area, thereby introducing ambiguity into the Defender’s belief. 
Specifically, if the Defender observes the Attacker on the red trajectory, its belief updates to $[0.264, 0.736]$—i.e., a 26.4\% chance that the Attacker is type-1 and 73.6\% for type-2. 
However, if the Attacker is observed on the magenta trajectory, the Defender infers with certainty that it is type-1 Attacker.

At the terminal states of these trajectories ($s_6$ for magenta and $s_6'$ for red), the values of the SAD-constrained MDP for each representation type and barrier configuration are shown below. The expected outcome is computed based on the Defender's belief $\vb_6$ on the Attacker's representation type.

\begin{minipage}{0.48\textwidth}
\vspace{+0.1in}
\centering
\renewcommand{\arraystretch}{1.2}
\begin{tabular}{@{}>{}p{2.7cm}| >{\centering\arraybackslash} p{2cm} >{\centering\arraybackslash} p{2cm}@{}}
$V^{\theta^*}_{\omega, \bm{s_6}} $ & ${\omega=1}$ & ${\omega=2}$ \\
\hhline{=|==}
${\theta=1}$ & -0.670 & 0.321 \\
${\theta=2}$ & 0.109 & -0.033 \\
\textbf{expected outcome} & -0.670 & 0.321
\end{tabular}
\vspace{+0.1in}
\end{minipage}
\hfill
\begin{minipage}{0.48\textwidth}
\vspace{+0.1in}
\centering
\renewcommand{\arraystretch}{1.2}
\begin{tabular}{@{}>{\bfseries}p{2.7cm}| >{\centering\arraybackslash} p{2cm} >{\centering\arraybackslash} p{2cm}@{}}
 $V^{\theta^*}_{\omega, \bm{s'_6}}$ & ${\omega=1}$ & ${\omega=2}$ \\
\hhline{=|==}
${\theta=1}$ & -0.800 & 0.093 \\
${\theta=2}$ & 0.095 & -0.225 \\
expected outcome & -0.141 & -0.141
\end{tabular}
\vspace{+0.1in}
\end{minipage}

If the Defender observes the magenta trajectory, it infers that the Attacker is type-1 and thus selects barrier configuration $\omega = 1$ to minimize the Attacker’s future performance.
In contrast, if the red trajectory is observed, the expected value of each barrier configuration—given the updated belief $[0.264, 0.736]$—is $-0.141$, making the Defender indifferent. 
As a result, the Defender randomizes between the two configurations with equal probability $[0.5, 0.5]$.

The induced uncertainty in the Defender's belief benefits the Attacker. 
For the type-1 Attacker, the Defender may select the suboptimal barrier~2, which allows the Attacker to take a straight-line path to the target (see Figure~\ref{fig:phase2-trjs}(c)).
Conversely, for the type-2 Attacker, the Defender may select barrier configuration~1, allowing the Attacker to enter the fine-resolution area directly (see Figure~\ref{fig:phase2-trjs}(d)).

\begin{figure}[b]
    \centering
    \includegraphics[width=\linewidth]{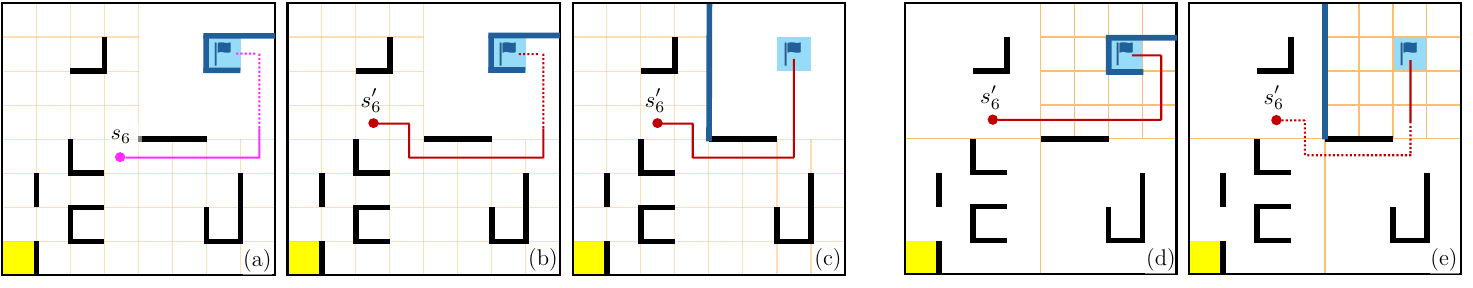}
    \caption{Trajectories in Phase~II. Subplots (a)–(c) show the trajectories under the type-1 representation, while (d)–(e) correspond to the type-2 representation. 
    The dashed segments indicate portions of the trajectory that require a randomized strategy; they represent the intended path, but the actual realized trajectory may differ due to stochasticity.}
    \label{fig:phase2-trjs}
\end{figure}

The heatmaps in \ref{fig:abstrt-conceal-heatmap} show the Phase~II value functions for each state under different barrier options and abstractions.
The color scale is consistent across subplots, with warmer colors indicating higher values.
Barrier option~1 yields the lowest values for abstraction type~1, while barrier option~2 is worst for abstraction type~2.
The chosen end states for Phase~I, $s_6$ and $s_6'$, are both of relatively high value.
However, the heatmaps do not reflect the difficulty of reaching these states given the abstraction constraints.
For example, although state $s_6$ has a high value for abstraction type~2 under barrier option~2, it is extremely difficult for the type-2 Attacker to reach due to randomized navigation with a coarse representation in the obstacle region.

\begin{figure}[t]
    \vspace{-0.1in}
    \centering
    \includegraphics[width=0.9\linewidth]{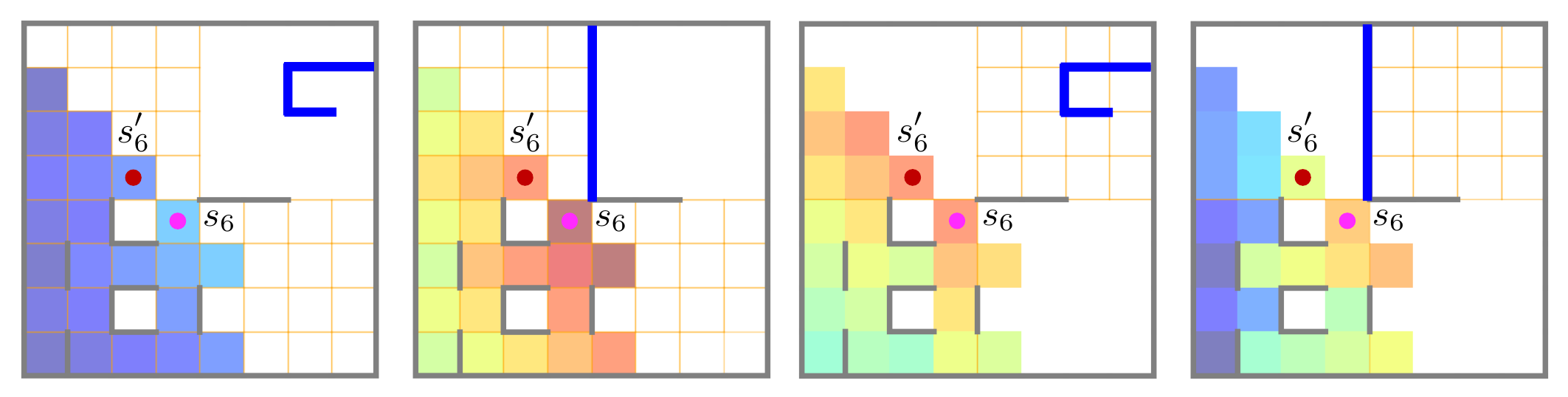}
    \caption{Value heatmaps of the Phase~II constrained MDP for states reachable within $T = 6$.}
    \label{fig:abstrt-conceal-heatmap}
    \vspace{-0.1in}
\end{figure}

To illustrate the benefit of representation concealment, consider first the 
\emph{full-information variant} of the game in which the Defender knows the 
Attacker’s representation type $\theta$ at the outset. 
In this case, the Defender deterministically selects the type-optimal barrier 
($\omega=1$ if $\theta=1$, $\omega=2$ if $\theta=2$), and the Attacker responds 
with the corresponding type-optimal trajectory (cyan path for type 1 and red path for type 2). 
The Attacker’s expected performance, averaged over the prior distribution on 
types, is $-0.978$, computed based on
\begin{equation*}
    -0.978 \;=\; 
    \underbrace{-0.530}_{\text{Phase~I}} 
    + \underbrace{0.5\times -0.670}_{\theta=1,\, \text{cyan path},\, \omega=1} 
    + \underbrace{0.5\times -0.225}_{\theta=2,\, \text{red path},\, \omega=2}.
\end{equation*}

In contrast, in the asymmetric-information setting the Attacker can employ 
randomized trajectories in Figure~\ref{fig:example-trajectory} to introduce ambiguity in Defender's belief as discussed. 
This raises the Attacker’s expected performance to $-0.841$:
\begin{align*}
    -0.841 &= 
    \underbrace{-0.530}_{\text{Phase~I}} 
    + \underbrace{0.5 \times 0.64 \times -0.67}_{\theta=1,\, \text{cyan path}\, \omega=1} 
    + \underbrace{0.5 \times 0.36 \times 0.5 \times -0.8}_{\theta=1,\, \text{red path},\, \omega=1} 
    + \underbrace{0.5 \times 0.36 \times 0.5 \times 0.093}_{\theta=1,\, \text{red path},\, \omega=2} \\
    &\quad 
    + \underbrace{0.5 \times 0.5 \times 0.095}_{\theta=2,\, \text{red path},\, \omega=1} 
    + \underbrace{0.5 \times 0.5 \times -0.225}_{\theta=2,\, \text{red path},\, \omega=2}.
\end{align*}
This corresponds to a $14\%$ performance improvement when both phases are included. 
If the constant Phase~I payoff is excluded, the relative performance improvement from representation concealment rises to $30.5\%$.

Finally, we examined the influence of the prior distribution $\vb_0$ on the Attacker's behavior.
When $\vb_0 = [0.8, 0.2]$, the type-1 Attacker adopts a similar randomization strategy as in the uniform prior case, randomizing at Point A with probabilities $0.09$ for \texttt{Up} and $0.91$ for \texttt{Down}, while the type-2 Attacker continues to follow the deterministic red path in Figure~\ref{fig:example-trajectory}. 
This yields the same posterior belief $[0.26, 0.74]$ at state $s_6'$, leaving the Defender indifferent between the two barrier placements.

In contrast, with $\vb_0 = [0.2, 0.8]$, both Attacker types deterministically follow the red path to $s_6'$, resulting in a posterior belief of $[0.2, 0.8]$. 
In this case, the Defender deterministically selects barrier configuration $2$. The rationale behind such a strategy is that reproducing the $[0.26, 0.74]$ posterior from the $[0.2, 0.8]$ prior would require the \emph{type-2} Attacker to randomize at Point A. 
Since the type-2 Attacker employs a coarse representation in Phase~I, randomization may cause the Attacker to be trapped near obstacles, resulting in a longer path in Phase~II. 
As a result, the algorithm assigns the deterministic red path to both types of Attacker.

\section{Conclusion}
This work explored the asymmetric information structure that arises when agents employ reduced environment representations in competitive games. 
We developed a Bayesian-game framework for the strategic concealment and inference of such representations in an attack–defense scenario, where the Attacker leverages its private representation to transform informational advantage into performance gain. 
Specifically, the Attacker must plan a trajectory to its goal while concealing its representation, whereas the Defender seeks to infer this hidden representation and select barriers to obstruct the Attacker’s advance.
We integrated the Defender’s inference and the Attacker’s planning into a unified bilinear optimization problem. Simulation results demonstrate that purposeful representation concealment naturally emerges as a means of improving the Attacker’s performance.

Future work includes extending the framework to settings where both sides operate with private representations, and developing methods to autonomously generate meaningful representation classes and priors, rather than assuming them as given.

 \bibliographystyle{ieeetr}
 \bibliography{ref.bib}

\newpage
\appendix

\section{Proof of Proposition 1}
\label{appdx-sec:prop-1}
\eqv*
\begin{proof}
    The ``$\Rightarrow$” direction is immediate and holds even when $b_t^\theta = 0$ for some $\theta$.
    
    For the ``$\Leftarrow$” direction, we proceed by contradiction.
    Assume that for all $\pi_{h_t} = \big(\pi_{h_t}^\theta\big)_{\theta \in \Theta} \in (\mathcal{P}(\A))^{|\Theta|}$, the following inequality holds:
    \begin{equation}
        \label{eqn:equivalence-assumption}
        \sum_{\theta} b_t^\theta J^\theta_t(h_t, \pi_{h_t}^{\theta*}) \geq \sum_{\theta} b_t^\theta  J^\theta_t(h_t, \pi_{h_t}^{\theta}).
    \end{equation}
    
    Suppose, on the contrary, that there exists some $\hat \theta \in \Theta$ and a corresponding Attacker strategy $\pi_{h_t}^{\hat \theta} \in \mathcal{P}(\A)$ such that 
    \begin{equation}
        \label{eqn:equivalence-contradiction}
        J^{\hat\theta}_t(h_t, \pi_{h_t}^{\hat \theta*}) < J^{\hat\theta}_t(h_t, \pi_{h_t}^{\hat \theta}).
    \end{equation}
    Consider the policy $\tilde \pi_{h_t} = (\pi_{h_t}^{1*}, \ldots, \pi_{h_t}^{\hat \theta}, \ldots, \pi_{h_t}^{|\Theta|}) \in \big(\mathcal{P}(\A)\big)^{|\Theta|}$ that modifies only the type-$\hat \theta$ component of the Attacker equilibrium policy, while keeping all other components unchanged. Then,
    \begin{align*}
        \sum_{\theta} b_t^\theta J^\theta_t(h_t, \pi_{h_t}^{\theta*}) - \sum_{\theta} b_t^\theta  J^\theta_t(h_t, \tilde \pi_{h_t}^{\theta})
        &= \sum_{\substack{\theta \in \theta \\ \theta\ne \hat{\theta}}} \Big(b_t^\theta J^\theta_t(h_t, \pi_{h_t}^{\theta*}) - J^\theta_t(h_t, \pi_{h_t}^{\theta*})\Big) + b_t^{\hat \theta}\Big( J^{\hat \theta}_t(h_t, \pi_{h_t}^{\hat \theta*}) - J^{\hat\theta}_t(h_t, \pi_{h_t}^{\hat \theta}) \Big) \\
        & = b_t^{\hat \theta}\big(J^{\hat \theta}_t(h_t, \pi_{h_t}^{\hat \theta*}) - J^{\hat\theta}_t(h_t, \pi_{h_t}^{\hat \theta})\big) < 0,
    \end{align*}
    where the strict inequality follows from~\eqref{eqn:equivalence-contradiction} and the assumption that $b_t^{\hat \theta} > 0$.
    This inequality contradicts the assumption~\eqref{eqn:equivalence-assumption}, thus completing the proof.
\end{proof}

\section{Bayesian Rules and Belief Update}
\label{appdx-sec:belief-update}
To derive the belief update rule from its conditional probability definition, we begin with the formal expression for the belief at time $t$:
\begin{align*}
    b_t^\theta &= \mathbb{P}_{\pi}\big(\Theta = \theta \big \vert H_t = h_t\big)
    = 
    \frac{\mathbb{P}_{\pi} (\Theta = \theta, H_t = h_t) }{\mathbb{P}_{\pi} (H_t = h_t)}
    =   \frac{\mathbb{P}_{\pi}\big(H_t = h_t \big \vert \Theta = \theta \big) \mathbb{P}(\Theta = \theta)}{\sum_{\theta'} \mathbb{P}_{\pi}\big(H_t = h_t \big \vert \Theta = \theta' \big)\mathbb{P}(\Theta = \theta')}\\
    &=\frac{b^{\theta}_0\prod_{\tau=0}^{t-1}\mathbb{P}\big(S_{\tau+1}=s_{\tau+1}|S_\tau = s_\tau, A_\tau = a_\tau\big) \pi_{h_\tau}^\theta (a_\tau)}{\sum_{\theta'} b^{\theta'}_0\prod_{\tau=0}^{t-1}\mathbb{P}\big(S_{\tau+1}=s_{\tau+1}|S_\tau = s_\tau, A_\tau = a_\tau\big) \pi_{h_\tau}^{\theta'} (a_\tau)}\\
    &=\frac{b^{\theta}_0\prod_{\tau=0}^{t-1}\mathbb{P}\big(S_{\tau+1}=s_{\tau+1}|S_\tau = s_\tau, A_\tau = a_\tau\big) ~\prod_{\tau=0}^{t-1}\pi_{h_\tau}^\theta (a_\tau)}{\sum_{\theta'} b^{\theta'}_0\prod_{\tau=0}^{t-1}\mathbb{P}\big(S_{\tau+1}=s_{\tau+1}|S_\tau = s_\tau, A_\tau = a_\tau\big) ~\prod_{\tau=0}^{t-1}\pi_{h_\tau}^{\theta'} (a_\tau)}.
\end{align*}

Since the transition probabilities do not depend on $\theta$, they cancel from numerator and denominator, yielding
\[
     b_t^\theta=\frac{b^{\theta}_0\prod_{\tau=0}^{t-1} \pi_{h_\tau}^\theta (a_\tau)}{\sum_{\theta'} b^{\theta'}_0\prod_{\tau=0}^{t-1}\pi_{h_\tau}^{\theta'} (a_\tau)}.
\]

For time $t+1$, the belief becomes
\[
     b_{t+1}^\theta = \frac{b^{\theta}_0\prod_{\tau=0}^{t} \pi_{h_\tau}^\theta (a_\tau)}{\sum_{\theta'} b^{\theta'}_0\prod_{\tau=0}^{t}\pi_{h_\tau}^{\theta'} (a_\tau)}.
\]

Define the normalizing factor as
\[\eta_{h_{t}} = \sum_{\theta' \in \Theta} b^{\theta'}_0\prod_{\tau=0}^{t-1}\pi_{h_\tau}^{\theta'} (a_\tau).\] 

Then, we obtain
\begin{align*}
    b_{t+1}^\theta &= \frac{b^{\theta}_0\prod_{\tau=0}^{t} \pi_{h_\tau}^\theta (a_\tau)}{\sum_{\theta'} b^{\theta'}_0\prod_{\tau=0}^{t}\pi_{h_\tau}^{\theta'} (a_\tau)}
    =\frac{\frac{b^{\theta}_0\prod_{\tau=0}^{t-1} \pi_{h_\tau}^\theta (a_\tau)}{{\eta_{h_t}}}~\pi_{h_t}^\theta(a_t)}{\sum_{\theta'}\frac{b^{\theta'}_0\prod_{\tau=0}^{t-1} \pi_{h_\tau}^{\theta'} (a_\tau)}{{\eta_{h_t}}}~\pi_{h_t}^{\theta'}(a_t)} 
    = \frac{b_{t}^\theta \pi_{h_t}^{\theta} (a_t)}{\sum_{\theta'} b_{t}^{\theta'} \pi_{h_t}^{\theta'} (a_t)},
\end{align*}
which matches the belief update rule in~\eqref{eqn:belief-update}.
\vspace{+1in}

\end{document}